\makeatletter\def\set@pdftextpagesize{\set@pdftexpagesize}\makeatother

\documentclass[
 paper=A4,pagesize=automedia,fontsize=12pt,
 BCOR=6mm,DIV=18,
 twoside,headinclude,footinclude=false,
 UKenglish,fleqn,             
 openany,
 bibliography=totoc,          
 listof=totoc,                
 listof=flat,                 
 numbers=noenddot,          
 cleardoublepage=empty      
]{scrbook}
\usepackage[utf8]{inputenc}
\usepackage[UKenglish]{babel}
\usepackage{cmap}
\usepackage[T1]{fontenc}
\usepackage{scrpage2} 
                      \clearscrheadfoot
                      \ihead{\headmark}\ohead{\pagemark}
                      \automark[section]{chapter}
                      \setheadsepline{0.5pt}
\usepackage{setspace} 
\deffootnote{1em}{1em}{\textsuperscript{\thefootnotemark }}
\usepackage{graphicx,xcolor}
\usepackage{tabularx}
\usepackage{amsmath,amsfonts,amssymb}
\usepackage{flafter,afterpage}
\usepackage[section]{placeins}
\usepackage[margin=8mm,font=small,labelfont=bf,format=plain]{caption}
\usepackage[margin=8mm,font=small,labelfont=bf,format=plain]{subcaption}

\usepackage{amssymb}
\usepackage{amsmath}
\usepackage{amsthm}
\usepackage{mathtools}
\usepackage{calrsfs}
\usepackage[justification=centering]{caption}
\usepackage{appendix}
\usepackage{multirow}

\newtheorem{theorem}{Theorem}
\newtheorem{lemma}{Lemma}

\makeatletter
\newcommand*{\textoverline}[1]{$\overline{\hbox{#1}}\m@th$}
\makeatother

\numberwithin{equation}{chapter}
\numberwithin{figure}{chapter}
\numberwithin{table}{chapter}

\makeatletter
\g@addto@macro\normalsize{%
  \setlength\abovedisplayskip{6pt}
  \setlength\belowdisplayskip{10pt}
  \setlength\abovedisplayshortskip{6pt}
  \setlength\belowdisplayshortskip{10pt}
}
\makeatother

\begin{document}

\frontmatter 

\begin{titlepage}
 \begin{tabularx}{\linewidth}{X}
  \vspace{4.5em}

  \begin{singlespace}\begin{center}\bfseries\Huge
  
  The Structure of Free Gauge Fields over Minkowski Space\\
  \vspace{1.5em}
  
  \Large Dedicated to Lars Gårding in Memoriam
  \vspace{3.5em}
  
  Christian Högfors Ziebell\\
  
  \large
  CB Centre for Biomechanics AG\\
  Hinter der Saline 9, D-21339 Lüneburg, Germany\\
  Chr.Ziebell@yahoo.de\\
  
  \end{center}\end{singlespace}
  
  \vspace{3.5em}
  
  \begin{center}\large\textbf{Abstract}\end{center}
It is shown that the physical states of a source free gauge field form pre-Hilbert spaces already on the classical level. These spaces may be closed in such a way that the determining characteristics remain. One type of free fields has a continuous mass spectrum with no mass gap. One type is massless. Dark matter may consist of (otherwise) free gauge fields.\\
Second quantisation of free gauge fields yield Gårding-Wightman fields identifying the massless and the massive fields as different manifestations of the same field.

 \end{tabularx}
\end{titlepage}

\cleardoublepage
\tableofcontents

\mainmatter 

\cleardoublepage
\chapter{Introduction}
The Gårding-Wightman framework for quantum field\footnote{A field in this context is a smooth geometrical object.} theory is a mathematically well-developed theory, which in a set of axioms tries to capture the general properties of quantum fields on physical space-time \cite{source:streater1964pct}, \cite{source:jost1965general}.\\
These axioms may be taken as

\begin{enumerate}
\item The states of a system form a complex separable Hilbert space $H$. The elements of $H$ and their inner products are to be interpreted according to relativistic quantum mechanics \cite{source:von1932mathematische} and reflect Wigner’s \cite{source:wig1959group} symmetry conditions.
\item There exist a finite number of $C^{\infty}$ test function spaces $F_k$ whose elements are geometrical objects on $\mathbb{R}^{1,3}$ and linear mappings $A_k : F_k \rightarrow L(H)$ where $A_k(f)$, $\forall \, f \in F_k$ have a common dense domain $D$.\\
The mappings $\left< \Phi, A_k(\cdot) \Psi \right> : F_k \rightarrow \mathbb{C}$ belong to $F_k^{'}$, $\forall \, \Phi \in H$ and $\forall \, \Psi \in D$.
\item For each $A_k$ there exists a (strongly continuous) representation $\pi$ of the semi-direct product of translations $T$ of $\mathbb{R}^{1,3}$ and the component of the Lorentz group connected to the identity $L_k^{\uparrow}$ or its covering group $SL(2, \mathbb{C})$. (The different $A_k$ may belong to different spin representations.) We let $G$ denote the actual group and $g \in G$ its corresponding elements.
\item There exists a unique vector $\Omega \in D$ that is invariant under $G$ (i.e. $\forall \, g \in G : \pi(g) \Omega = \Omega$). The Stone measure $E$ associated with the (abelian) subgroup $T \subset G$ has support in the closure of the forward light cone of $\mathbb{R}^{1, 3}$.
\item If $f_1 \in F_i$ and $f_2 \in F_j$ have space-like separated supports, then $A_i(f_1)$ and $A_j(f_2)$ commute (or anti-commute) on $D$.
\item The image of the ring of polynomials in $A_k(f)$, $f \in F_k$, acting on $\Omega$ is dense in $H$.
\end{enumerate}

For a motivation for this set of axioms we refer to \cite{source:von1932mathematische} and \cite{source:wig1959group}. However reasonable these assumptions may seem, there are some pitfalls.
All known massive quantum field models on Minkowski space $\mathbb{R}^{1,3}$ are models of (generalised) free\footnote{The term 'free' is actually not well-defined quantum mechanically. Mostly it refers to classically source free differential equations.} fields.
\emph{Massive} means (c.f. axiom 4) $\text{supp}(E(p)) = \{p=0\} \cup \{p \cdot p \ge m^2\}$ for $p \in \mathbb{R}^{1,3}$ and some real non-zero $m$. Free is a term indicating that the modelled field is obtained through the procedure of canonical quantisation of a classical field without interaction \cite{source:dirac1932the}.
A slight modification of the properties of known quantum fields leads to the concept of \emph{generalised free fields} \cite{source:fock1932konfigurationsraum} also belonging to the set of models obeying the axioms.
Generalised free fields may be obtained by canonical quantisation of classical fields, whose Lagrangian densities are quadratic. (Note that canonical quantisation is not well-defined except in these two cases.)
The (separable) Hilbert space is then a Fock space. (As a matter of fact reference \cite{source:fock1932konfigurationsraum} has had a strong influence on the formulations of the axioms.)
There was hitherto, however, no known example of a quantum field theory with interaction fulfilling the axioms in $\mathbb{R}^{1,3}$. Another difficulty is presented by Haag’s theorem \cite{source:haag1955quantum}, raising several questions about the fourth axiom.\\
The Gårding-Wightman requirements on a quantum field theory do not touch upon the question of “quantisation” of a classical theory. Several schemes are available, none is quite satisfactory.\\

There is also an alternative to the functional analytic approach of Gårding and Wightman, which is an algebraic framework proposed by Haag and Kastler \cite{source:haag1964algebraic}. Although the approaches are quite different, some very similar results are derivable in both frameworks.\\
The scope of this paper is rather modest in view of the cited references. The question was raised, as one of the Millennium Problems from the Clay Mathematics Institute, whether a free quantum Yang-Mills theory \cite{source:milleniumYMmassgap} exists and if such a theory may exhibit a mass gap. The somewhat imprecise problem formulation was:\\
“Prove that for any compact simple gauge group $G$, a non-trivial quantum Yang–Mills theory exists on $\mathbb{R}^4$ and has a mass gap $\Delta > 0$. Existence includes establishing axiomatic properties at least as strong as those cited in \cite{source:streater1964pct} and \cite{source:jost1965general}“.\\
This paper affirms only parts of the proposition. There exists, for any smooth group $G$, a non-trivial quantum field theory on an arbitrary principal bundle with connection over Minkowski space and structure group $G$ satisfying the Gårding-Wightman requirements and exhibiting massive as well as massless states.
The field considered is the curvature form of the connection. When the curvature form is source free, there is no mass gap.\\
In this paper the structure of a gauge theory, classical or quantum, is elucidated. The term \emph{gauge group} stems from the Lagrangian formulation of field theory, classical or not.
If a Lagrangian is invariant under the action of a compact Lie group on the (components of the) field, we may “localise” this invariance.
This means that instead of splitting the base space into manifolds where the group elements may be taken as constants, we consider the fields to be cross-sections of a fibre bundle associated to the principal bundle over the base, having this Lie group as structure group.

This is classically just a change of description (coordinate transformation on the total space).\\
We note that, although the axiomatic framework for quantum field theory is not dependent on a Lagrangian approach, nor on variational methods, we base much of the terminology on such approaches.
Thus a gauge group is not always arbitrary, e.g. when the above-mentioned fibre bundle must allow for integration of non-gauge-invariant objects.
Similarly the quantum Hilbert space is most often an $L^2$ type space, also requiring integration. When all physical objects are taken to be gauge invariant, any smooth group is possible.\\
Whenever gauge symmetry is broken, possible gauge groups are likely to be restricted to the special unitary groups (and their subgroups).\\
Today a \emph{Yang-Mills theory} \cite{source:yang1954conservation} denotes a gauge theory with the bundle connection considered as a field. Such a field has a quadratic Lagrangian density.\\
In physics a field theory is termed \emph{trivial} either if it is equivalent to a (generalised) free theory, or if its associated $S$-matrix \cite{source:heisenberg1943beobachtbaren} is trivial.
To define an $S$-matrix the theory must allow for \emph{asymptotic states}. These may be defined in a model satisfying the Gårding-Wightman axioms but require some additional assumptions (Haag-Ruelle theory) \cite{source:jost1965general}. The added assumptions do however not suffice to exhibit triviality or non-triviality of the $S$-matrix.\\
The main result of this paper is:
\begin{itemize}
\item If the curvature field $\Omega$ of any fibre bundle associated with a principal bundle over $\mathbb{R}^{1,3}$  with smooth structure group $G$ satisfies $D \star \Omega = 0$ , then $\Omega$ is a finite sum of lifted null field solutions of Maxwell’s equations. 
\item The one particle states in a (second) quantised description are of two kinds, one of which carries mass zero and the other with a continuous mass spectrum. The two sectors are not orthogonal.
\item The free fields associated with gauge fields satisfy the Gårding-Wightman axioms.
\end{itemize}

In chapter \ref{chapter:electromagnetism} the structure of the free electromagnetic field is investigated. It is shown that the smooth states of a free field constitute two pre-Hilbert spaces, which may each be closed retaining its desired properties.
In chapter \ref{chapter:yangmills} analogous methods show that a free Yang-Mills field has a similar structure, exhibiting the same two-field structure as described by Maxwell’s equations, one massless and one massive.
In the appendix some condensed background material and notation is presented.
The paper is concentrated on the particular aspects concerning smooth principal bundles over Minkowski space.
Most detail and general theorems involved can be found in standard textbooks for mathematical physicists \cite[Choquet-Bruhat et al.]{source:choquet1977analysis}, \cite[Dubrovin et al.]{source:dubrovin1984modern}, \cite[Abraham et al.]{source:abraham1988manifolds}, \cite[Taylor]{source:taylor1986noncommutative}.

\chapter{A Digression on the Electromagnetic Field}
\label{chapter:electromagnetism}

\section{Some Mathematical Consideration}
\label{section:mathconsideration}
We adhere to the physics convention in that we call a smooth section of a bundle a field.
The speed of light is set to 1. In physics a classical electromagnetic field is modelled as a harmonic $2$-form $F$, on Minkowski space $M = \mathbb{R}^{1,3}$, i.e. the abelian topological group $\mathbb{R}^4$ with its translation invariant topology, together with a metric tensor $\eta$ of signature $(1, -1, -1, -1)$.\\

In coordinates the strict (covariant) components of $F$ are $F = \sum \limits_{\mu < \nu} F_{\mu \nu} dx^{\mu} \wedge dx^{\nu}$ and we identify the components of the matrix $F_{\bullet \bullet}$ with components of the electric and magnetic fields (in suitable units) on $\mathbb{R}^3$ as follows:

\begin{equation*}
F_{\bullet \bullet} = 
\begin{pmatrix}
0 & E_1 & E_2 & E_3 \\
-E_1 & 0 & H_3 & -H_2 \\
-E_2 & -H_3 & 0 & H_1 \\
-E_3 & H_2 & -H_1 & 0 \\
\end{pmatrix}
\end{equation*}

Maxwell’s (microscopic) equations for the electromagnetic field read:
\begin{align}
\frac{\partial H_1}{\partial x^1} + \frac{\partial H_2}{\partial x^2} + \frac{\partial H_3}{\partial x^3} &= 0 
\label{eq:maxwellDivH} \\
\frac{\partial E_1}{\partial x^1} + \frac{\partial E_2}{\partial x^2} + \frac{\partial E_3}{\partial x^3} &= 4 \pi \rho 
\label{eq:maxwellDivE}
\end{align}

\begin{align}
\frac{\partial (H_1, H_2, H_3)}{\partial x_0} &=
  \left( \frac{\partial E_3}{\partial x^2} - \frac{\partial E_2}{\partial x^3}, 
   \frac{\partial E_1}{\partial x^3} - \frac{\partial E_3}{\partial x^1},
   \frac{\partial E_2}{\partial x^1} - \frac{\partial E_1}{\partial x^2} \right) 
\label{eq:maxwelldHdt} \\
\frac{\partial (E_1, E_2, E_3)}{\partial x_0} &=
  - \left( \frac{\partial H_3}{\partial x^2} - \frac{\partial H_2}{\partial x^3}, 
   \frac{\partial H_1}{\partial x^3} - \frac{\partial H_3}{\partial x^1},
   \frac{\partial H_2}{\partial x^1} - \frac{\partial H_1}{\partial x^2} \right) + 4 \pi \rho \cdot (v_1, v_2, v_3) 
\label{eq:maxwelldEdt}
\end{align}

Or equvalently with equations \ref{eq:maxwelldHdt}, \ref{eq:maxwellDivH} $\Leftrightarrow$ \ref{eq:maxwelldF} and equations \ref{eq:maxwelldEdt}, \ref{eq:maxwellDivE} $\Leftrightarrow$ \ref{eq:maxwellSdSF}:
\begin{align}
\text{d} F &= 0 \label{eq:maxwelldF} \\
\star^{-1} \text{d} \star F &= 4 \pi j \label{eq:maxwellSdSF}
\end{align}

In these equations $j = \rho \cdot (1, v_1, v_2, v_3)$ where $\rho$ is the electric charge density and $v = (v_1,v_2,v_3)$ is the 3-velocity of charge. If $\rho(x) = 0, \forall \, x \in \mathbb{R}^{1,3}$ we deal with a source free field, a \emph{free field} for short.
In terms of the physical fields on $\mathbb{R}^3$, the electric field $E = -(E_1, E_2, E_3) = (F^{01}, F^{02}, F^{03})$ and the pseudo-vector $H = (H_1, H_2, H_3) = (F^{23}, -F^{13}, F^{12})$ representing the magnetic field, we may write Maxwell’s equations in their conventional vector analysis form:

\begin{align*}
\text{div} \; E &= 4 \pi \rho & \frac{\partial H}{\partial x^0} &= \text{curl} \; E \\
\text{div} \; H &= 0 & \frac{\partial E}{\partial x^0} &= - \text{curl} \; H + 4\pi \rho v \\
\end{align*}

These equations are invariant under the component of the Lorentz group connected with the identity.
The pointwise eigenvalues of $F$ relative to the Lorenz metric are the roots of

\begin{equation*}
\lambda^2 = \frac{E^2 - H^2}{2} \pm \sqrt{\frac{(E^2 - H^2)^2}{4} + (E \cdot H)^2}
\end{equation*}

Wells \cite{source:wells1973differential} called solutions of Maxwell’s equations whose pointwise eigenvalues are zero everywhere \emph{null solutions}.\\
The Hodge star operation on the $2$-form $F$ is an involution: $\star \star F = -F$. This is often used to define a complex structure on the space of $2$-forms at a point on $M$ by $\forall z \in \mathbb{C} : z F = \text{Re} \, z \cdot F + \text{Im} \, z \cdot \star F$, i.e. the space of $2$-forms at a point of $M$ is isomorphic to $\mathbb{C}^3$.\\
We may thus employ the monomorphism $\phi : \Lambda^2(TM \rvert_{t = t_0}) \rightarrow T \mathbb{C}^3$ such that
$F_\phi := \phi(F) = E + i H$. (The vector $\frac{1}{\sqrt{2}}(E + i H)$ in $\mathbb{C}^3$ is sometimes called the Riemann-Silberstein vector). This vector is the eigenvector of the Hodge star operator in the $\phi$ representation.\\
In the $\phi$ representation null solutions are elements of the set in $\mathbb{C}^3$ given by $F_\phi \cdot F_\phi = 0$.\\
When the components of $E$ and $H$ belong to $D(\mathbb{R}^3)$ for every $t$, then the "free energy” integral 

\begin{equation*}
\int \limits_{\mathbb{R}^3} (E(t, x) - i H(t, x)) \cdot (E(t, x) + i H(t, x)) \text{d}^3 x =
  \int \limits_{\mathbb{R}^3} \langle F_\phi, F_\phi \rangle \text{d}^3 x
\end{equation*}

 is defined and is independent of $t$ (since $\partial_t \langle F_\phi, F_\phi \rangle =  \langle \text{curl} \; F_\phi, F_\phi \rangle + \langle F_\phi, \text{curl} \; F_\phi \rangle$).
 Here $\langle \cdot, \cdot \rangle$ denotes the standard Hermitean scalar product on $\mathbb{C}^3$ and $\int \limits_{\mathbb{R}^3} \langle \cdot, \cdot \rangle \text{d}^3 x$ is the corresponding $L^2$ norm, where the adjoint of the curl operator is minus itself.\\
 
Maxwell’s equations without sources constitute a constant coefficient homogenous hyperbolic system.
A solution is thus uniquely determined by initial conditions and the regular set of smooth solutions, $S_R = \{F_\phi, G_\phi, \dots \}$ form a complex linear space with the positive sesquilinear form, scalar product, given by $\int \limits_{\mathbb{R}^3} \langle F_\phi, G_\phi \rangle \text{d}^3x$.
This space may be closed to a separable Hilbert space, $H_F$, in the associated $L^2$ type norm. (I will below exhibit a null base for this Hilbert space.) We will accordingly at first analyse possible initial conditions.
As we are doing physics we will concentrate on initial conditions having bounded energy, i.e. belonging to $L^2$. The cases of compact regular and singular support will be treated separately.\\

We define the support of the initial condition as the support of the field $F(0, x)$.
First we will assume that the field has regular compact support $S$ and that $E(0, x)$ and $H(0, x)$ have components belonging to $D(\mathbb{R}^3)$.
The source free Maxwell equations state that $\text{div} \; E = \text{div} \; H = 0$, so we are dealing with divergence free $SO(3)$ vectors with compact support.
Being vector fields on domains with compact closure, both $E$ and $H$ are complete.
Since $\text{div} \; E = \text{div} \; H = 0$ in $\text{Int}(\text{supp}(E))$ and $\text{Int}(\text{supp}(H))$ respectively (mutatis mutandi), these vector fields have no zeroes on their respective domains.
That the components belong to $D(\mathbb{R}^3)$ implies that they are bounded.
$E$ and $H$ are thus non-singular on their domains, implying the existance of smooth functions $e$ and $h$ such that $E(e) = H(h) = 0$.
The linear space of divergence free vector fields whose components belong to $D(\mathbb{R}^3)$ will be denoted $D^\nu(S)$ or $D^\nu$ when its support is understood.\\
The following theorem is central to the analysis and physical understanding of the free smooth electromagnetic field and other gauge fields.\\

\begin{theorem}
\label{theorem:sumofsourcefreenullsolutions}
Any source free solution to Maxwell’s equations with initial vectors ($E$ and $H$) belonging to $D^\nu$ may be uniquely written as a sum of source free null solutions.
\end{theorem}

The theorem will be proved through a sequence of lemmas.

\begin{lemma}
\label{lemma:sumoftwonullsolutions}
Any solution to Maxwell’s equations with initial data in $D^\nu$ may be written as a sum of (at most) two null solutions. For any $t_0$ the splitting is unique on $\text{supp}(F(t_0, x)) \cap \text{supp}(E(t_0, x) \times H(t_0, x))$.
\end{lemma}

\begin{proof}
The proof is simply by construction. The field may be represented by $F = E + i H$, where $E$ and $H$ as well as their spatial Fourier transforms, $\overline{E}$ and $\overline{H}$, are vector fields on $\mathbb{R}^3$ under orientation preserving coordinate transformations. In the generic case one may write $F = F_1 + F_2$, where

\begin{align*}
F_1 &= E - \frac{E \times H}{\lvert E \times H \rvert} \times H + i\left(H + \frac{E \times H}{\lvert E \times H \rvert} \times E \right) \\
F_2 &= E + \frac{E \times H}{\lvert E \times H \rvert} \times H + i\left(H - \frac{E \times H}{\lvert E \times H \rvert} \times E \right) \\
\end{align*}
                         
Real and imaginary parts of these $F$’s are orthogonal at every point and have the same absolute value there. As a matter of fact

\begin{align*}
\frac{E \times H}{\lvert E \times H \rvert} \times \text{Im}(F_1) &= -\text{Re}(F_1) \\
\frac{E \times H}{\lvert E \times H \rvert} \times \text{Im}(F_2) &= \text{Re}(F_2) \\
\end{align*}

Should $E(t_0, x_0) \times H(t_0, x_0)$ be zero at the point $(t_0, x_0)$ but different from zero in a spatial neighbourhood $N(t_0, x_0) = N(x_0)$ of $x_0$, $F_1$ and $F_2$ may be extended by continuity.
If $E(t_0, x_0) \times H(t_0, x_0)$ is zero for all $x \in N(x_0)$, a splitting is trivial on $N(x_0)$.\\
However, the parts perpendicular to the (parallel) field vectors $E$ and $H$ have an arbitrary direction (that may be used to (non-uniquely) interpolate between points on the boundary between $\text{supp}(F(t_0, x))$ and $\text{supp}(E(t_0, x) \times H(t_0, x))$).
\end{proof}

\begin{lemma}
\label{lemma:divergencefree}
When two vector fields belonging to $D^\nu$, $E$ and $H$, on $\mathbb{R}^3$ are in involution on $S$, i.e. both are tangent to the same foliation of $S$, then the fields $F_1$ and $F_2$ constructed in the proof of lemma \ref{lemma:sumoftwonullsolutions} are divergence free.
\end{lemma}

\begin{proof}
Wherever the vector fields are parallel (or either is zero) in an open subdomain $T$ of $S$ the lemma is trivial on this subdomain.
Otherwise the expression

\begin{equation*}
H_{\perp} := \frac{(E \times H) \times H}{\lvert E \times H \rvert}
\end{equation*}

is a vector field everywhere orthogonal to $H$ and of the same magnitude.
According to Frobenius theorem the assumptions imply that $H_{\perp}$ is tangent to the $2$-dimensional integral manifold $N$ of $E$ and $H$.
Thus the mutually orthogonal vector fields $H$, $H_{\perp}$ together with the normal to $N$ span the domain $\text{Int}(S \setminus T)$.
We denote by a lower index $l$ the one form metrically corresponding to a vector field. Then the maximal form $\omega := H_l \wedge H_{\perp l} \wedge \text{d}N$ is a volume form, since $E_l$ and $H_l$ never vanish.\\
Now obviously $i_{H_{\perp}} \omega = 0$, which implies $\text{div}(H_{\perp} ) = 0$. 
\end{proof}

\begin{lemma}
\label{lemma:maximalclosedsetinvolution}
Let $E$ and $H$ belong to $D^\nu(S)$ and $S$ be compact in $\mathbb{R}^n$. Assume further that $E$ and $H$ are in involution on a subset $B \subset S$ with $B$ not of measure zero.
Then there exists a maximal closed set $C$, $B \subset C \subset S$, such that $E$ and $H$ are in involution on $C$ and not in involution except on sets of measure zero in $S \setminus C$.
\end{lemma}

\begin{proof}
The assumptions imply the existence of smooth functions $e$ and $h$ such that $L_E(e) = 0$ and $L_H(h) = 0$.
$E$ and $H$ make up an integral system on $B$, so we may choose $e$ and $h$ such that $\text{d}e = \text{d}h$ on $B$.
Then there exists a non-extendible set $A$ such that $\forall p \in A : \text{d}e(p) = \text{d}h(p)$.
The closure of the interior of $A$ is the set $C$.
Sard’s theorem ascertains that $e$ and $h$ intersect transversally almost everywhere in $S \setminus C$.
\end{proof}

\begin{lemma}
\label{lemma:existsinvolution}
Given two elements, $E$ and $H$ of $D^\nu(S)$ with $S$ compact in $\mathbb{R}^n$, then there exists an element $W \in D^\nu(S)$ such that $W$ is in involution with both $E$ and $H$.
\end{lemma}

\begin{proof}
The proposition is trivial for $n = 1,2$. For $n \ge 3$ the assumptions imply that that the vector fields, $E$ and $H$ are complete.
Accordingly there exist smooth functions $e, h : \mathbb{R}^n \rightarrow \mathbb{R}$ such that $E \cdot \text{d}e = H \cdot \text{d}h = 0$.
In case one can find $e$ and $h$ such that $\text{d}e = \text{d}h$ on some open submanifold of $S$, then there is a largest closed set $C \subseteq S$, such that $\text{d}e = \text{d}h$ on $C$ according to lemma \ref{lemma:maximalclosedsetinvolution}.
Both $E$ and $H$ are tangent to $\partial C$ from both sides.
On $C$ there is nothing to prove since either $E$ or $H$ may be taken as $W$ (restricted to $C$).
On $S \setminus C$, $E$ and $H$ intersect transversally almost everywhere.\\
Consider the integral curves of any smooth vector field $X$ contained in the $(n-2)$-dimensional (except on sets of zero measure) intersections of $e+k$ and $h+l$ in $\text{Int}(S \setminus A)$, where $k$ and $l$ denote constants.
Let $W$ be any smooth vector field having the same integral curves as $X$. $W$ is defined everywhere in $S$, $W$ is complete since $L_W(e) = L_W(h) = 0$ and $W$ preserves the volume element in $\text{Int}(S \setminus A)$ since both $E$ and $H$ are tangent to $\partial C \cup \partial S$. 
\end{proof}

\begin{proof}[Proof of Theorem \ref{theorem:sumofsourcefreenullsolutions}]
Assume we are given a compactly supported $C^\infty$ source free solution to Maxwell’s equation.
$S$ is the support of this solution at $t = 0$. Consider the boundary of $S$.
Parts of $\partial S$ may enclose open connected sets. Since $S$ is compact with (piecewise) smooth boundary and the solution is smooth, the number of such open connected sets $\Omega_i \in S$ is countable.
The field restricted to $\Omega_i$ consists of two divergence free vector fields on $\mathbb{R}^3$, which we will call $E$ and $H$.
According to lemma \ref{lemma:existsinvolution} we may write the field as a sum of two fields in involution:\\

\begin{equation*}
F = E + i H = [E + iW + i(H + iW)] + [E - iW + i(H - iW)]
\end{equation*}

These two fields may be separately orthogonalised according to lemmas \ref{lemma:sumoftwonullsolutions} and \ref{lemma:divergencefree}.
Thus $F \rvert_{\Omega_i}$ may be written as a sum of at most four null fields.
\end{proof}

A null field as defined in theorem \ref{theorem:sumofsourcefreenullsolutions} will subsequently be called a \emph{fundamental entity}.\\

The above theorem describes possibilities to decompose $C^\infty$ source free compactly supported electromagnetic fields into pairs of $C^\infty$ compactly supported orthogonal divergence free vector fields in involution.
However already on this classical level, we may recognise a classification problem regarding fundamental entities (FE) i.e. the elements of the Hilbert space $H_F$ obtained by completing the space spanned by $C^\infty$ compactly supported orthogonal divergence free vector fields in involution.
We have here chosen a seemingly natural one, induced by the requirement on the pre-Hilbert space elements that the fields have $C^\infty$ components and regular compact support $S$.
This forces the traces of the FE's to be zero on $\partial S$. Thus a FE, i.e. an element of $H_F$ is characterised as follows:

\begin{enumerate}
\item The FE’s have rest systems. As a matter of fact they are at rest in any frame that is coupled to the rest frame through a Lorentz transformation. In general they also have energy and may be considered to have mass = energy in a rest system. Disregarding the conserved orientation, they transform according to a massive representation of $L_+^\uparrow$.
\item It has a connected compact support and its interior is also connected. In other words it is a bounded connected domain for two orthogonal, equal magnitude, smooth and divergence free vector fields. The support is independent of time in the rest frame.
\item At time $t = 0$ (or at any fixed time) the interior of the support is foliated. There is one exact form annihilating both vector fields. The foliation is not preserved in time.
\item The $L^2$ norm of the vector fields (the energy) is independent of time. So is any spatial moment of the energy.
\item The orientation of the vector fields (relative to the normal to the foliation) is independent of time.
\item The components of the vector fields are harmonic functions on $\mathbb{R}^{1, 3}$.
\item The FE’s are covariant with respect to the Lorenz group.
\end{enumerate}

Since Maxwell’s equations are linear, possibly inhomogeneous, we may consider solutions to them as a sum of an arbitrary solution to the homogenous equation and some particular solution. We have just shown that a solution to the homogenous equation with $C^\infty(\mathbb{R}^3)$ initial functions having compact support with non-empty interior is an element of the Hilbert space spanned by the FE’s with the $L^2$ norm.\\
Now clearly the source free Maxwell equations will have solutions even if the initial functions are not $C^\infty(\mathbb{R}^3)$.
The prime examples are given by the plane waves. Thus we are led to consider singularly supported initial values.
When at $t = 0$ the initial vectors are supported on a two-dimensional surface $N$ in $\mathbb{R}^3$, the condition that both $E$ and $H$ are divergence free means that these vectors both are tangent to $N$.
Tangent vectors are the only ones preserving the volume form on $N$.
Here we have assumed that the singularly supported vectors are $C^\infty$ on $N$.
They are clearly in involution and may be split into a pair of null solutions, each with its own orientation (helicity).
What singular surfaces $N$ can be supporting solutions to the homogenous equations?
$N$ is a two-sided wave-front set. If $N$ has curvature it cannot at all times support a homogenous solution.\\
Thus plane waves are the only singularly supported solutions to Maxwell equations.\\
Homogenous as well as inhomogeneous solutions obviously satisfy $(\text{d}\delta + \delta\text{d})F = 0$, i.e. they are harmonic.
This means that a complex vector field whose components satisfy the wave equation is a solutions to Maxwell’s equations.
Furthermore if both the real and imaginary parts are divergence free it is a solution to the homogenous equations.\\
Consider compactly singularly supported solutions that are $C^\infty$ in directions along $N$, i.e. whose components belong to $D(N)$.
These solutions are dense in a Hilbert space $H_P$ with norm $\int \limits_N \langle F_\phi, F_\phi \rangle \text{d}^2x$ where the components of the complex vector $F_\phi$ now are singularly supported on $N$.
Obviously these solutions may as in the smooth case be resolved into null fields, which we will call \emph{photons}.\\

Denote a solution belonging to the “massive” Hilbert space $H_F$, the completion of the space of FE’s with initial values in $D^\nu(S)$ in the $L^2$ norm of vector fields on $\mathbb{R}^3$, by $\phi$ and a solution belonging to the massless “wave front” Hilbert space $H_P$ by $\psi$.
Note that $H_F$ and $H_P$ have no elements in common (except for the zero vector), since $\psi \notin L^2(\mathbb{R}^3)$ and $H_F$ is the completion in $L^2(\mathbb{R}^3)$ of smooth FE’s.
Thus both of these Hilbert spaces must be considered in the quantisation procedure. It is a special feature of $\mathbb{R}^3$ that the solutions of free gauge field equations are limited to $H_F$ and $H_P$.\\

Solutions to the inhomogeneous Maxwell equations, $\text{d}F = 0; \; \delta F = j$ (with $\text{d}j = 0$ as a condition for a solution to exist) may be obtained as usual by considering \emph{fundamental solutions}.
If in a coordinate system a distribution $E$ satisfies Maxwell’s equations in the form $M(D)E = \delta I$, where $M(D)$ is a $6 \times 6$ matrix of differential operators (c. f. the explicit form at the beginning of this chapter) and $I$ is a unit matrix, then $E$ is a fundamental solution.
By standard results $E$ has support in a cone \cite{source:hormander1983analysis} (here $t > 0$). We denote the fundamental solutions to Maxwell’s equations by $E_+$ and $E_-$ with supports in  $t > 0$ and $t < 0$ respectively.
Thus $F_p = E_+ \ast j$ is a solution for $t > 0$. However a (distributional) solution to the homogenous equation is given by $F_h = E_+ - E_-$.
These expressions are not particularly transparent, but clearly display the difference between free and non-free solutions.\\

So far we have established that the classical linear solution space to the source free Maxwell equations consists of the two Hilbert spaces $H_F$ and $H_P$, both of which exhibit particle like characteristics.
In order to describe the geometry of the state space we need, apart from scalar products within $H_F$ and $H_P$, the scalar product between elements of the two spaces.\\
$H_P$ is spanned by objects supported on planes of form $n \cdot x - t = 0$ where $n$ is the unit normal to $N$ and $t \in \mathbb{R}$.
Consider a compactly supported element $\phi$ of $H_F$ in its rest frame with support $S_\phi$.
For some compact interval $I$ of the real line, $t \in I \subset \mathbb{R}$, $N \cap S$ is nonempty.
Chose an element $\psi$ of $H_P$ with support $S_\psi \subset N$.
Whenever $S_\phi$ and $S_\psi$ have a non-empty intersection, the formal expression

\begin{align}
\label{eq:mixedScalarProduct}
\int \limits_{\mathbb{R}^3} \langle \phi | \psi \cdot \delta(N) \rangle \text{d}^3 x &=
\int \limits_{\mathbb{R}^3} \langle \phi | \psi \cdot \delta(n \cdot x - t) \rangle \text{d}^3 x =
\int \limits_{\mathbb{R}^3} \langle \delta(N) (\phi - (n \cdot \phi) n) | \psi \rangle \text{d}^3 x \nonumber \\
&= \int \limits_{N} \langle (\phi - (n \cdot \phi) n) \rvert_N | \psi \rangle \text{d}^2 x
\end{align}
 
is bounded and different from zero. We denote this object tentatively as $\langle \phi || \psi \rangle$ and note that from its definition it satisfies $\lvert \langle \phi || \psi \rangle \rvert \le \| \phi \| \| \psi \|$.
Thus the space of solutions to the homogeneous Maxwell equations form a pre-Hilbert space with elements $\alpha \phi + \beta \psi$, where $\alpha$ and $\beta$ are complex numbers and $\phi \in H_F$, $\psi \in H_P$.
This solution space may be completed to a Hilbert space $H$ in the norm $\| \phi + \psi \|_H = \| \phi \|_{H_F} + \| \psi \|_{H_P} + 2 Re \langle \phi \, || \, \psi \rangle$.\\
It is furthermore clear from equation \ref{eq:mixedScalarProduct} and Riesz's lemma that there exists in $H_F$ a unique element $\xi$ corresponding to the functional $\psi \delta(N)$,
i.e $\langle \phi \, || \, \psi \delta(N) \rangle = \langle \phi \, || \, \xi \rangle$.
Similarly there of course exists in $H_P$ a unique element corresponding to the functional $(\phi - (n \cdot \phi) n) \rvert_N$.\\
Canonical quantization \cite{source:berezin1966method} of the free Maxwell theory thus leads to the symmetric Fock space over $H$, i.e $\oplus_0^\infty S H^{\otimes n}$, as the state space (for the number representation).
We then arrive at a free field theory satisfying the Gårding-Wightman requirements.
Note however that the number operator in the second quantised theory counts the sum of FE’s and photons.
Should we insist on a physical two-field interpretation, i.e should it be possible to count FEs and photons separately, such a Gårding-Wightman theory were nontrivial.

\section{Where is the Physics?}
The solutions to the homogenous Maxwell equations presented in section \ref{section:mathconsideration}. can, from a mathematical standpoint, hardly be regarded as strange.
Physicists with an experience in field theory, classical or quantum, should at least find the FE’s rather exotic.
These entities may be regarded at rest in any Lorentz frame, in particular in the rest system of any massive “particle”.
Maxwell’s equations are first order and the solutions are accordingly uniquely determined by the conditions at an initial time.
A FE is essentially determined by its support at a particular time, its orientation property and its energy.
All these properties, including the support, are then unchanged in time in the coordinate system where they were determined.
“Viewed” from a photon a FE “looks like” another photon.
We could argue that a FE hides from massive observers.\\
Quantum field theory requires the propagator, the difference between forward and backward fundamental solutions, to be “observable”.
As I am considering a free theory, the question of what is observable of course becomes rather philosophical.
However, if we consider the state space of compactly supported solutions (in the space variables) to the homogenous Maxwell equations, such a state is described by superposition of FE's and photons.\\

A FE has the same energy (rest-mass) in any orthonormal coordinate system of $\mathbb{R}^{1,3}$. Its centre of energy can be set to be the origin in any system. However it’s shape, i.e. it’s support does change. When boosted a FE shrinks in the direction of the boost and the volume diminishes. The energy density increases, retaining a constant energy. This remains true even if the speed of the boost approaches $1$. A FE remains a particle-like object, having mass and a centre of mass. In contrast a photon possesses a centre of energy. It has no mass, but it carries momentum.\\
In order to be able to verify the physical existence of FE’s, we would either have to use gravitation theory or be able to observe what happens locally, when a charged particle (or better a pair of oppositely charged particles) comes in the vicinity of a FE. \\
Should we consider the FE’s and the photons to be physically different fields, we somehow have to live with the non-vanishing “transition probability”\footnote{The proper quantum mechanical interpretation is of course that neither the quantum number operator for the FE’s, nor that for the photons, (were they defined) has a static expectation value. This is true only for the sum of these operators.} between elements of $H_F$ and $H_P$.
This would be an example of a non-trivial Gårding-Wightman theory.
The continuous spectrum of the unitary representation of the translation operator makes it differ though from hitherto presented models.

\emph{Speculations:} The cosmologists claim that there is in our galaxies a deficit of visible mass. In interstellar space, far away from particle sources and where net gravitational effects are small, one could imagine FE’s to be sizeable. Localised electromagnetic energy could thus possibly explain the prevalence of “dark matter”. Since to my knowledge FE’s are not observed as stable objects on earth, one may speculate that there are somehow too many particles around that make these objects unstable through influence from charges or microscopic curvature variations. Should this be the case, we have to expect FE’s on earth to be very low energy phenomena. Should FE’s not be observable, and dark matter actually be associated to some other source, we were forced to conclude that Maxwell’s equations are not quite correct when gravity is considered. Finally, if the existence of physical FE’s is demonstrable, these should play a role in cosmological entropy considerations. The entropy of a FE seems at first sight to be larger than the entropy of the more singular photon.

\chapter{Free Yang-Mills Fields}
\label{chapter:yangmills}
Yang-Mills fields are fields identified with the curvature form on a principal fibre bundle with a smooth structure group $G$. The situation here is somewhat different from the electromagnetic field treated in the earlier section. Maxwell’s equations (equations \ref{eq:maxwelldF} and \ref{eq:maxwellSdSF}), were derived from physical observations only. The formalism exhibits that in the Maxwell case one may regard the two-form $F$ as a curvature form of a $U(1)$ bundle over Minkowski space, turning equation \ref{eq:maxwelldF} into an identity. Equation \ref{eq:maxwellSdSF} is then an equation of motion derived through Lagrangian formalism. The Lagrangian density is here taken to be the (trace of) the curvature squared. In a remarkable paper \cite{source:yang1954conservation} Yang and Mills proposed an analogous setup for an isospin multiplet.\\
As in section \ref{section:mathconsideration} I will present the results in terms of vector fields. It is assumed that the principal fibre bundle $(P, \mathbb{R}^{1,3}, G, \pi, \{U_k, \phi_k\})$ is equipped with a connection form $\omega$ or, equivalently, with a G-connection $\sigma$. For some purposes it is advantageous to regard $\omega$ as a collection of one-forms, the number of which equals $m := \text{dim}(G)$. For the notation employed see the appendix. $\omega$ annihilates horizontal vectors. Through the connection a metric tensor $\Gamma$ is defined on $P$,

\begin{equation*}
\Gamma(X,Y) = \Gamma(X_h + X_\nu, Y_h + Y_\nu) = \Gamma(X_h,Y_h) - K(X_\nu, Y_\nu).
\end{equation*}

The particular horizontal vector field defined as the lift of $\partial_t = \frac{\partial}{\partial t}$, (where the coordinates on $\mathbb{R}^{1,3}$ are taken as $(t,x)$ with $x$ an element of $\mathbb{R}^3$) may be integrated to a one-dimensional foliation orthogonal to the submanifold hypersurfaces $t = t_0 \; \forall t_0 \in \mathbb{R}$. The Bianchi identity locally takes the form

\begin{equation*}
\text{D}\Omega = \sum \limits_{\lambda < \mu < \nu} (\nabla_\lambda F_{\mu \nu} + \nabla_\nu F_{\lambda \mu} + \nabla_\mu F_{\nu \lambda}) \text{d}x^\lambda \wedge \text{d}x^\mu \wedge \text{d}x^\nu = 0
\end{equation*}

where

\begin{equation*}
F_{\mu \nu} = \frac{\partial A_\nu}{\partial x^\mu} - \frac{\partial A_\mu}{\partial x^\nu} + [A_\mu, A_\nu]
\end{equation*}

and $A$ is a Lie algebra valued one-form on a local trivialisation $(\pi^{-1}U_k, U_k, G, \pi, \{\phi_k\})$.
The relation between $A$ and $\omega \rvert_{\pi^{-1}U_k}$ is

\begin{equation*}
\omega(x, g) = \omega_{MC}(g) + g A_\mu(x) g^{-1} \text{d}x_k^\mu
\end{equation*}

The indices of the form $F$ are raised and lowered with the inverse of the metric tensor on the base.\\
Picking up the thread from chapter \ref{chapter:electromagnetism} we now consider the (free) field $F = \Omega$ as a Yang-Mills field satisfying a Bianchi identity $\text{D}F = 0$ and an Euler-Lagrange equation $\text{D} \star F = 0$.
The star operation is defined here through the Minkowski metric on the base $\mathbb{R}^{1,3}$.
As in chapter \ref{chapter:electromagnetism} we write the covariant and the contravariant tensors respectively as

\begin{equation*}
F_{\bullet \bullet} = 
\begin{pmatrix}
0 & E_1 & E_2 & E_3 \\
-E_1 & 0 & H_3 & -H_2 \\
-E_2 & -H_3 & 0 & H_1 \\
-E_3 & H_2 & -H_1 & 0 \\
\end{pmatrix}
\qquad
F^{\bullet \bullet} = 
\begin{pmatrix}
0 & -E_1 & -E_2 & -E_3 \\
E_1 & 0 & H_3 & -H_2 \\
E_2 & -H_3 & 0 & H_1 \\
E_3 & H_2 & -H_1 & 0 \\
\end{pmatrix}
\end{equation*}

where now the “vectors” $E =  -(E_1, E_2, E_3)$ and $H = (H_1, H_2, H_3)$ are Lie algebra valued.
These objects transform as vectors under $SO(3)$. The equations $\text{D}F = 0$ and $\text{D} \star F = 0$ are each $m$ equations.
These equations are homogenous and hyperbolic and thus determined by the restriction of $F$ to the submanifold $t = 0, F \rvert_{t = 0}$. 

\begin{theorem}
Any solution to the equations $\text{D}F = 0$ and $\text{D} \star F = 0$, where $F$ is the curvature form on a smooth principal bundle $(P, \mathbb{R}^{1,3}, G, \pi, \{U_k, \phi_k\})$ over Minkowski space with smooth structure group $G$, projection $\pi$ and a smooth G-connection $\sigma_p$, such that the Lie algebra components of the initial vectors $E_0 := E(0, x)$ and $H_0 := H(0, x)$ belong to $D(\mathbb{R}^3) \times G$, may be uniquely written as a finite sum of lifts of source free null solutions on the base.
\end{theorem}

\begin{proof}
It is enough to show that $F_0$ may be written in such a way.
In coordinates $F = \prescript{i}{}{F}_{\mu \nu} \text{d}x^\mu \text{d} x^\nu l_i$, where $l_i$ may be any local basis for the Lie algebra at $g = e$.  Then

\begin{align*}
\text{D} F &= 0 && \Rightarrow & \sum \limits_{\alpha = 1}^3 \nabla^\alpha H_\alpha &= 0 \\
\text{D} \star F &= 0 && \Rightarrow & \sum \limits_{\alpha = 1}^3 \nabla^\alpha E_\alpha &= 0 \\
\end{align*}

These equations are valid for all $t$ and are thus also restrictions on the values of $H$ and $E$ at $t = 0$.
Consider the $m = \text{dim}(G)$ horizontal vectors $\prescript{i}{}{E}^\alpha \nabla_\alpha$ and similarly $\prescript{i}{}{H}^\alpha \nabla_\alpha$ where $i \in \{1, ..., m\}$.
The assumptions imply that each of these vectors have divergence free projections

\begin{equation*}
\text{div} \prescript{i}{}{E}_\pi = \text{div}( \pi' \prescript{i}{}{E}^\alpha \nabla_\alpha) = 0
\qquad
\text{div} \prescript{i}{}{H}_\pi = \text{div}( \pi \prescript{i}{}{H}^\alpha \nabla_\alpha) = 0
\end{equation*}

According to Theorem \ref{theorem:sumofsourcefreenullsolutions}, $\prescript{j}{}{F}_\pi = \prescript{j}{}{E}_\pi + i \prescript{j}{}{H}_\pi$ may be written as a finite sum $\prescript{i}{}{F}_\pi = \prescript{i}{}{F}_\pi^1 + \prescript{i}{}{F}_\pi^2 + ...$ where $\forall n : \prescript{i}{}{F}_\pi^n \cdot \prescript{i}{}{F}_\pi^n = 0$.\\
Then $F$ may be written as

\begin{equation*}
F(p) = \sum \limits_{i = 1}^m l_i \sigma_p( \prescript{i}{}{F}_\pi(x) )
\end{equation*}

where $p \in \pi^{-1}x$.
\end{proof}

It is evident that the quantum version of a free gauge theory has exactly the same characteristics as it has in the Maxwell case, since horizontal fields are isometric to their projections on the base.\\
\emph{Comment:} The special properties of $\mathbb{R}^3$ are as important here as in the Maxwell case.\\
\emph{Discussion:} It was noted already in the introduction that the term “free field”, used in the characterisation of examples of Gårding-Wightman theories, is not well defined in quantum theories. The same of course goes for the term “trivial”, whenever there is not enough structure for an S-matrix to be defined.\\
The model systems investigated in this paper are “free fields” in the classical sense.
The fields are classically just the homogeneous solutions to linear hyperbolic PDE's.
The particle aspect of fields is a quantum phenomenon and in the second quantised version of a gauge field the massless and the (continuously) massive parts have to represent different kinds of “particles”.
However, from a quantum mechanical point of view, these “particles” represent different states of the same field.
In the traditional Fock representations the “one particle states” for different kinds of “particles” are orthogonal and in this sense is also the number operator diagonal, a situation obviously not prevailing in a quantised gauge theory.
I would suggest that we use the term “free fields” for fields that classically are represented by solutions to homogeneous linear PDE's. Similarly I propose the use of the descriptor “trivial theory” to be reserved for theories where the S-matrix exists and equals unity. (The existence of asymptotic states is not enough to guarantee the existence of an S-matrix.)

\begin{appendices}

\cleardoublepage
\addchap{Appendix}
\renewcommand\thesection{\Alph{section}}
\numberwithin{equation}{section}

\section*{Some (condensed) Background Material}
Gauge fields and in particular Yang-Mills fields are defined on fibre bundles over Minkowski space $\mathbb{R}^{1,3}$. The considered fields are then the curvature forms or their potentials, the connection forms, of a principal fibre bundle. Thus the classical electromagnetic field is identified with the curvature form of a principal $U(1)$ bundle over $\mathbb{R}^{1,3}$. Similarly curvature or connection forms on smooth principal bundles over $\mathbb{R}^{1,3}$ are termed Yang-Mills fields. Below follows a short overview of the concepts involved.\\
A smooth \emph{principal fibre bundle}, $\mathcal{P} : (P, B, G, \pi, \{U_k, \phi_k\})$ consists of the following data:\\
$P$, called the total space, and $B$, the base space, are smooth $(C^\infty)$, Hausdorff manifolds.
$B$ is para-compact. $G$ is a smooth Lie group having a smooth free left action $L : G \times P \rightarrow P$, $L(g, p) := L_g p = g(p)$ on $P$.
The space of orbits $P / G = B$ is the base space and the mapping $\pi : P \rightarrow P / G = B$ is a smooth surjective map with a Jacobian of everywhere maximal rank, $\text{dim}(B)$.
For any $x \in B$, $\pi^{-1}(x)$ is a smooth submanifold of $P$ isomorphic to $G$. $\{(U_k, \phi_k) : k \in J \subseteq \mathbb{N}\}$ is a family of open sets $\{U_k\}$ covering $B$ together with diffeomorphisms $\phi_k : \pi^{-1} U_k \rightarrow U_k \times G$ (local trivialisations) such that $\pi \circ \phi_k^{-1} : U_k \times G \rightarrow U_k$ is the projection on the first factor.
Furthermore, these diffeomorphisms define smooth mappings $g_k : \pi^{-1}U_k \rightarrow G$ by $\phi_k(p) = (\pi(p), g_k(p))$ for all $p \in P$.
These mappings are equivariant, i.e. $g_k(R_g(p)) = g_k(p)g_k^{-1}$ for all $p \in P$ and all $g \in G$.\\
If $U_k \cap U_l : k \neq l$ is nonempty, two, in general different, trivialisations over this region exist.
Since for any $p \in P$ both $g_k(p)$ and $g_l(p)$ are elements of $G$, there exist group elements $t_{kl}(p) = g_l(p)g_k^{-1}(p)$.
Due to the equivariance of $g_k(p)$, the group elements $t_{kl}(p)$ are constant along the fibres and the \emph{transition functions} $t_{kl} : U_k \cap U_l \rightarrow G$ are local sections.
(Global sections of a bundle exist if and only if $\mathcal{P}$ is homeomorphic to a product bundle.
Thus there exist sections on $(\pi^{-1} U_k, U_k, G, \pi, \{\phi_k\})$ called \emph{local sections}.
In particular let $s_k : U_k \rightarrow G, \, s_k(b) \rightarrow e$, then $\Sigma_k : U_k \rightarrow \pi^{-1} U_k, \, \Sigma_k(b) = \phi_k^{-1}(b, s_k(b)) = \phi_k^{-1}(b, e)$ is a local section and $\pi \circ \Sigma_k = \text{id} : B \rightarrow B$.)\\

The definition of $\pi$ allows the definition of \emph{vertical vectors} as those belonging to the kernel of $\pi'$.
The left free action of $G$ on $P$ allows the construction of \emph{left invariant vertical vector fields} as $\lambda_X(p) = \frac{\text{d}}{\text{d}t}(L_{e^{tX}} (p))$ for all fixed $X \in \mathbf{g}$, the Lie algebra of $G$.
We define \emph{right invariant vector fields} as $\rho_X(p) = -\frac{\text{d}}{\text{d}t}(R_{e^{tX}} (p))\rvert_{t = 0}$.
(The minus sign is slightly un-aesthetic, but has some advantages. Consider $\iota : G \rightarrow G$ by $\iota(g) = g^{-1}$. $\iota' : TG \rightarrow TG$ preserves Lie brackets and $\iota' \lambda_X = \rho_X$.
We may thus identify the Lie algebra of $G$ with the right as well as the left invariant vector fields.)
The common notation $e^{t X} := V(t)$ denotes the one-parameter subgroup $V(t) \in G$, defined on some interval containing $0$, with $V(0) = e$, the unit of $G$, satisfying $V(t)^{-1} \frac{\text{d} V(t)}{\text{d} t} \rvert_{t = 0} = X \in \mathbf{g}$.\\
Both left invariant and right invariant fields are vertical.
The \emph{fibre} $\pi^{-1}(b)$ may for any $b \in B$ be endowed with a measure by choosing a non-vanishing maximal ($\text{dim}(G)$) form $\beta$ at some arbitrary point $p_0 \in \pi^{-1}(b)$.
For $p = R_g p_0$ define $\beta(p) = -R_{g^{-1}}^* \beta(p_0)$, making $\beta$ a right invariant, nowhere vanishing maximal form on $\pi^{-1}(b)$ defining an orientation on each fibre.\\
A \emph{connection} (G-connection) $H$ on a principal fibre bundle $\mathcal{P}$ is a field of smooth mappings $\sigma : P \times TB \rightarrow TP$ with $\sigma(p, \cdot) := \sigma_p : T_{\pi(p)} B \rightarrow T_p P$, fulfilling the requirements that for each $p \in P$:

\begin{enumerate}
\item $\sigma_p$ is linear
\item $\sigma_p$ depends smoothly on $p$
\item $\pi' \sigma_p$ is the identity map on $T_{\pi(p)} B$ (normalisation)
\item $R_g' \sigma_p = \sigma_{R_g p}$
\end{enumerate}

The elements of the vector subspace $H_p = \sigma_p(T_{\pi(p)}(B)) \subset T_p(P)$ are called horizontal.
$H_p$ is complementary to the vertical subspace $V_p := \text{ker} \pi'(p)$ of $T_p(P) = H_P \oplus V_p$.
A vector $v \in TP_p$ is uniquely split into a horizontal part, $v_h = \sigma_p(\pi' v)$, and a vertical part, $v_v = v - \sigma_p(\pi' v)$.
Similarly the vector fields (differential operators) on $P$ obtained by “lifting” vector fields on $B$ may be termed horizontal.
(Due to the demands on $\sigma_p$ horizontal fields are uniquely determined $\forall p \in \pi^{-1} \text{supp}(v)$ where $v \in TB$.)
In a chart these horizontal fields give rise to “covariant derivatives”.
Note that the horizontal fields do not form an algebra.
The choice of an $H$ may also be specified by the smooth non-vanishing Lie algebra valued \emph{connection one-form} on $P$ annihilating $H_p$ for all $p \in P$.\\

The \emph{connection one-form} or a \emph{differential-geometric connection} $\omega$ on the principal bundle $\mathcal{P}$ has the following properties:

\begin{enumerate}
\item The restriction of $\omega$ to a fibre, $G$, yields the \emph{Maurer-Cartan one-form} or \emph{fundamental form} $\omega_{MC}$ on the fibre.
\item Under the natural left action of $G$ on $P$ the connection one-form satisfies $L_g^* \omega = \text{Ad}(g) \omega$.
\end{enumerate}

At a point $p \in P$, $\omega_p$ acts on a vector field $v$, with $v(p) \in TP_p$, by annihilating the horizontal part $v_h(p)$ of $v(p), \, \omega_p(\sigma_p(\pi' v(p))) = \omega_p(v_h(p)) = 0$.
Acting on the vertical part $v_v(p) = v(p) - v_h(p) = \chi_p(X(v_v(p)))$ of $v(p)$, the connection one-form produces the Lie algebra element $X_v := \omega_p(v_v(p)) = \chi_p^{-1}(v_v(p))$, where we identify the Lie Algebra with $TG_e$.
Note that $R_g^* \omega_p(v_v(p + v_h(p)) = R_g^* \omega_p(v_v(p)) = \omega_p(R_g'(\chi_p(X))) = \text{Ad}_{g^{-1}} X$.\\
We may regard $\omega \in \Lambda^1(P, \mathbf{g})$, a Lie algebra valued $1$-form on $P$ acting on $TP$, as an extension of $\chi^{-1} \in \Lambda^1(F \cong G)$ acting on the vertical fields of $TP$ only. As with $\sigma_p$, $\omega_p$ determines the same $H_p$ (and thereby $V_p$) as does $f(\pi p) \sigma(p)$ with $f \in C_d^\infty(B)$.
The normalisation of $\omega_p$ is the condition that $\omega_p(v_v(p)) = \chi_p^{-1}(v_v(p))$.\\
Connection one-forms exist on any principal bundle with paracompact base.
The space of connections on a given principal bundle is an affine subspace of the smooth, right invariant one-forms on $\mathcal{P}$.
When $\omega^1$ and $\omega^2$ are connection forms on $\mathcal{P}$, so is $\lambda \omega^1 + (1-\lambda) \omega^2$, where $\lambda \in C_d^\infty (\pi P)$.\\
In a local trivialisation $\phi_k : \pi^{-1} U_k \rightarrow U_k \times G$ the restriction of $\omega$ to $\pi^{-1} U_k$ takes the form $\omega_k(b, k) = \omega_{MC}(b) + \text{Ad}(g) A(b)$, where $A$ is a Lie algebra valued one-form on $U_k \subset B$.\\
The exterior covariant derivative $\text{D} \alpha$ of a Lie algebra valued $r$-form $\alpha$ on $\mathcal{P}$ is $\text{D} \alpha(v_1, ..., v_{r+1}) = \text{d} \alpha( \sigma \pi' v_1, ..., \sigma \pi' v_{r+1} )$.
Let $\omega$ be a connection one-form on $\mathcal{P}$, then $\Omega(u, v) := \text{D} \omega(u, v) = \text{d} (u, v) + [\omega(u), \omega(v)]$ and $\text{D} \Omega = 0$.\\

Since $B$ is paracompact it may be given a proper Riemannian structure. If $B$ is also non-compact it admits a hyperbolic structure.
If $G$ is semi-simple its Killing form is negative definite, thus under such circumstances a connection may obviously be used to extend (pseudo)-Riemannian structures to $P$.
More precisely let $\gamma$ be a (pseudo)-Riemannian metric on $B$, $K : TG \times TG \rightarrow R, \, K(X,Y) = \text{Tr} (\text{ad}(X) \cdot \text{ad}(Y))$ the Killing form on $G$ and $\omega$ a connection one form on $P$, then a symmetric bilinear mapping $\Gamma : TP \times TP \rightarrow R$ given by

\begin{equation*}
\Gamma(X, Y) = \Gamma(X_h + X_v, Y_h + Y_v) = \gamma(X_h, Y_h) \pm K(X_v, Y_v)
\end{equation*}

is a non-degenerate (pseudo)-Riemannian metric on $P$ when the structure group is semi-simple.
Moreover, horizontal and vertical vectors on $P$ are orthogonal under $\Gamma$.
If the minus sign is chosen in the definition of $\Gamma$ it defines a hyperbolic structure on $P$ whenever $B$ has one.
Note also that the form $\overline{\omega_0} = \Gamma(X_v, \cdot)$ is a scalar valued one-form annihilating the horizontal subspace for any $X_v$.
In particular it is also a basis for the vertical vectors, i.e. the invariant vector fields defining the Lie algebra of $G$.\\

The sections discussed above are assumed to be $C^\infty$, but we have not yet given a topology to $C^\infty$ sections, which is necessary for functional analysis.
On a general fibre bundle we may give compactly supported sections a topology induced by a family of semi-norms:\\
Let $\pi : P \rightarrow M$ be a smooth $N$-dimensional vector bundle of rank $n$, $U \subset P$ open and contractible, $\phi : U \rightarrow \mathbb{R}^N$ an admissible chart and $K \subset U$ compact.
Furthermore let $\{e_i\}$, $i = 1, ..., n$, be smooth local sections of $U$ such that $\{e_i(p)\}$ span the fibre at any $p \in U$.
Then for any section $s$ of the vector bundle and $R \in \mathbb{N}$:

\begin{equation*}
s_U = s^i e_i \quad \text{and} \quad \{\text{sup}_{p \in K} \text{sup}_{\lvert m \rvert \le R} \text{sup}_{i \le n} \left| \frac{\partial^{\lvert m \rvert} s^i}{\partial x^m} (p) \right| \},
\end{equation*}

where $x = \pi p$, is a semi-norm dependent on the compact set, the integer $R$, the chart and the choice of base sections.
This family of semi-norms defines a Frechét topology on the smooth sections of $P$ ($\Gamma^R(P)$ or when $R$ is unrestricted $\Gamma^\infty(P)$).
This $C^\infty$ topology is independent of the choices of charts and base sections and compactum.\\
Furthermore, compactly supported sections, $\Gamma_0^\infty (P)$, are dense in $\Gamma^\infty (P)$ in this topology.

\end{appendices}

\bibliography{paper} 
\bibliographystyle{ieeetr}

\end{document}